\documentclass{amsart}
\usepackage{amsmath, amssymb, amsthm, mathrsfs, tikz, enumitem, wasysym, mathtools,color,xcolor}
\usepackage{latexsym, hyperref}
\usepackage[linesnumbered,ruled,vlined]{algorithm2e}
\usepackage{parskip} 
\usetikzlibrary{arrows}
\newcommand{\commentout}[1]{}

\newcommand{\set}[1]{\{  #1 \} }

\newtheorem{theorem}{Theorem}[section]
\newtheorem{claim}[theorem]{Claim}

\newtheorem{defn}[theorem]{Definition}
\newtheorem{fact}[theorem]{Fact}

\newtheorem{remark}[theorem]{Remark}         
\commentout{\theoremstyle{plain}
\newtheorem{defn}{Definition}
\newtheorem{claim}{Claim}
\newtheorem{fact}{Fact}

\theoremstyle{remark}

}

\begin{document}

\title{\textbf{The $k$-Dimensional Weisfeiler-Leman Algorithm}}
\author{Neil Immerman}
\address{University of Massachusetts, Amherst, MA, USA}
\email{\href{mailto:immerman@cs.umass.edu}{immerman@cs.umass.edu}}

\author{Rik Sengupta}
\address{University of Massachusetts, Amherst, MA, USA}
\email{\href{mailto:rsengupta@cs.umass.edu}{rsengupta@cs.umass.edu}}
\maketitle

\begin{abstract}
In this note, we provide details of the $k$-dimensional Weisfeiler-Leman Algorithm and its analysis from \cite{canon}. In particular, we present an optimized version of the algorithm that runs in time $O(n^{k + 1}\log n)$, where $k$ is fixed (not varying with $n$).
\end{abstract}

\section{Introduction}

For $k \in \mathbb{N}$, the $k$-dimensional Weisfeiler-Leman algorithm (henceforth referred to 
simply as the \emph{WL-algorithm} or the \emph{$k$-WL algorithm}) takes as input an undirected graph,
colors all $k$-tuples of its vertices, and then iteratively refines the color classes based on a
generalized notion of ``colored neighbors''. We can use this as an isomorphism test, by applying the
$k$-dimensional WL-algorithm to the disjoint union of graphs $G$ and $H$.  Assume for simplicity
that $G$ and $H$ are connected graphs.  If for some $k$, the set of stable colors of $k$-tuples of
vertices from $G$ is disjoint from the set of stable colors of $k$-tuples of
vertices from $H$, then we say that the $k$-WL algorithm \emph{distinguishes} $G$ and
$H$. It is well known that if $k$ is the smallest integer such that the $k$-WL algorithm
distinguishes graphs $G$ and $H$, then $k + 1$ is the smallest number of variables in first-order
logic with counting that distinguishes $G$ and $H$ \cite{canon}. In particular, $\mathcal{C}^{k +
  1}$-equivalence 
corresponds to the $k$-WL algorithm. Furthermore, we know that if two graphs are distinguished by
the $k$-WL algorithm for some $k$, then they are certainly not isomorphic; the converse, while
\emph{often} true \cite{BK}, is not \emph{always} true \cite{CFI}. 

On an $n$-vertex graph $G$, the $k$-WL algorithm terminates with its output a stable coloring $\chi_G^k$ after at most $O(n^{k + 1})$ rounds, where the stable coloring corresponds to the coloring at the first round where the color classes remain unchanged. The color of a $k$-tuple after $m$ iterations of the WL algorithm is exactly the set of properties of that $k$-tuple expressible in $C^{k+1}_m$, i.e., counting logic with $k+1$ variables and quantifier-rank $m$:

\begin{fact} [{\cite{canon}}]For any graph, $G$, all $k,m$ and any two $k$-tuples of vertices of $G$, $\vec{a}, \vec{b}$, the following conditions are equivalent:

\begin{enumerate}
\item $\prescript{}{m}\chi^k_G(a_1, \ldots, a_k) = \prescript{}{m}\chi^k_G(b_1, \ldots, b_k)$
\item $(G, a_1/x_1, \ldots, a_k/x_k) \equiv_{C^{k+1}_m} (G, b_1/x_1, \ldots, b_k/x_k)$.
\end{enumerate}
\end{fact}

An optimization of the $k$-WL algorithm 
runs in time $O(n^{k + 1}\log n)$, where $k$ is fixed (not varying with $n$). This therefore corresponds to the time
required to check $\mathcal{C}^{k + 1}$-equivalence \cite{canon}.
The purpose of this note is to describe this algorithm and its analysis.

\section{The One-Dimensional Algorithm}

\subsection{Description of the Algorithm}

When $k = 1$, the WL-algorithm is simply known as the \emph{color refinement algorithm}. The input
is an undirected, uncolored graph $G = (V, E)$, and the output is the \emph{coarsest} stable coloring of $G$. We present the algorithm below.

\IncMargin{1em}
\RestyleAlgo{boxruled}
\begin{algorithm}\label{alg:WL1}
\SetAlgoLined
\SetKwInOut{Input}{Input}\SetKwInOut{Output}{Output}
\Input{An uncolored, undirected graph $G = (V, E)$.}
\Output{The \emph{stable} coloring of $G$.}

\vspace*{.1in}

\textbf{Initialization:} $C[v] = 1$, for all $v\in V$; $M = \varnothing$; $L = \set{1}$.  

** All vertices initially colored 1; multiset $M$ empty; work list $L$ initialized with color 1.

\vspace*{.1in}

\While{$L \neq \varnothing$}{
  \For{each vertex color class, $c \in L$}{remove $c$ from $L$.

    \For{each vertex $w$ with $C[w] = c$}{
      \For{each neighbor $v$ of $w$}{add $(v,c)$ to $M$.}
     }
   }

Perform Radix Sort of $M$.

Scan $M$ replacing tuples $(v,c_1), \ldots, (v,c_r)$, with single tuple: $(C(v), c_1, \ldots, c_r,v)$.

Perform Radix Sort of $M$.

Scan $M$: for each color class $c$ that has been split, leave the largest part still colored $c$,
and update the colors of the other parts of $c$; add these new colors to $L$; $M= \varnothing$.
}
Output $G$ with its current coloring. 
\caption{The $1$-dimensional Weisfeiler-Leman Algorithm}
\end{algorithm}

For avoiding clunkiness in explanations, let us state a quick definition.
\begin{defn}
During any fixed round of Algorithm \ref{alg:WL1}, define the \emph{$L$-vertices} to be the set of
all vertices in the color classes currently on the work list $L$; define the \emph{$L$-edges} to be
the set of all edges in $G$ with an $L$-vertex as an endpoint.
\commentout{Finally, say the round (or any subroutine within a round) is \emph{$L$-linear} if it takes time linear in the number of $L$-vertices plus the number of $L$-edges.}
\end{defn}

Note that now it makes sense to talk about the \emph{$L$-neighbors} of a vertex $v \in V$ during a fixed round of Algorithm \ref{alg:WL1}:  this is simply the set of all $w \in N(v)$ such that $w$ is an $L$-vertex.

\begin{remark}
It is important to understand how the work list $L$ is updated during each round of the
algorithm. Each color class present in $G$ after the $i$th round is either preserved or split during
the $(i + 1)$st round. If a color class is preserved, we do not include it in $L$; if it is split,
we let the largest part retain its previous color, and include \emph{all the other parts} in
$L$. For instance, suppose $\{v_1, v_2, v_3, v_4\}$ had color $t$ after the $i$th round, and suppose
the $(i + 1)$st iteration splits them into $\{v_1, v_2\}$, $\{v_3\}$, and $\{v_4\}$. Being the
largest part of the split, the vertices $\{v_1, v_2\}$ retain their old color $t$, and we add the
new colors of $v_3$ and $v_4$ to $L$.  In particular, $L$ keeps track of the \emph{split} color classes, and so $L$ being nonempty after some round corresponds to at least one color class being (strictly) refined during that round.
\end{remark}

\begin{remark}
There are two sorting steps in Algorithm \ref{alg:WL1}, in lines 12 and 14. They have different roles. The sort in line 12 is indexed by the vertices $v$, and so for each $v$ it clumps together the tuples $(v,C[w])$ for all of $v$'s $L$-neighbors $w$. This sort therefore labels a vertex $v$ with its number of $L$-neighbors of each color. Combined with the old color of $v$, this determines the new color of $v$.  

The sort in line 14 is indexed by the \emph{old color classes}, and so for each old color class it clumps together all $L$-neighbors that used to be in this old class. This now enables us to \emph{count} the sizes of the new color classes, in order to
determine which (if any) color classes have been split, so that we may update $L$. Note that radix
sort of a sequence of strings over the alphabet $\set{1, \ldots, n}$ takes time $O(\ell_{\rm{total}} +
n)$, where $\ell_{\rm{total}}$ is the total length of the strings; see Theorem 3.2 of \cite{AHU}.
Let $r$ be the number of $L$-vertices during a given round. Since there are fewer than $n$ edges from
any vertex, the strings being sorted and processed in lines 12-15 have total length
$\ell_{\rm{total}}\leq O(rn)$.  As we will see, it thus follows that the whole round including the
two sorting steps takes time at most $O(rn)$.
\end{remark}

\subsection{Proof of Correctness and Runtime}
Let us see now why this optimized version of the $1$-dimensional WL-algorithm is correct and
efficient.
We prove a running time of $O(n^2\log n)$, which is sufficient for our 
purposes.  A different implementation with the tighter bound  $O((m + 
n)\log n)$ (where $m$ is the number of edges) appears in \cite{BBG}.

\begin{claim}
Algorithm \ref{alg:WL1} terminates with a stable coloring of $G$.
\end{claim}
\begin{proof}
The work list $L$ keeps track of all color classes refined during the previous iteration. Every round where the algorithm does not terminate, therefore, corresponds to a strict refinement of some color class. Therefore, the process does indeed terminate eventually.

The condition for termination is that the work list $L$ is empty. But observe that $L$ can only be
empty if in the current round no color class splits. Thus, the output is the desired stable coloring.
\end{proof}

\begin{claim}\label{claim:logn}
The color class of any vertex $v \in V$ can appear in $L$ at most $1 + \log n$ times.
\end{claim}
\begin{proof}
Suppose we are at the end of the $i$th round of the algorithm. For any vertex $v \in V$, its color class will appear in $L$ during the $(i+1)$st round only if $v$'s color class was just split and $v$ is \emph{not} in the largest piece of the ensuing partition.  Thus, each time $v$'s color class appears on $L$, this class is at most half the size it was during the previous round.
\end{proof}

\begin{claim}\label{claim:runtime}
On input $G = (V, E)$ with $|V| = n$, Algorithm \ref{alg:WL1}  runs in time $O(n^2\log n)$.
\end{claim}
\begin{proof}
  We show that round $i$ of the main while-loop (lines 3-16) can be implemented to run in time $O(r_in)$, where $r_i$ is the number of $L$-vertices in round $i$.
The $i$th round starts by cycling through all $L$-vertices $w$, and scanning their adjacency lists to
update the multiset $M$ with their neighbors $v$. The size of $M$, therefore, is at most $O(r_in)$.
The total length $\ell_{\rm{total}}$ of the strings sorted in the two radix sorts are at most
$O(r_in)$, so the radix sorts take time at most $O(n + \ell_{\rm{total}})$ $=$ $O(r_in)$.
The scanning and
renaming step in line 13 is also linear in the size of $M$, since all the tuples $(v, c_1), \ldots, (v,
c_r)$ starting with a particular $v$ will appear consecutively after the sort in line 12.

Line 15 describes the process of reassigning the colors.  All the $L$-neighbors that had been color $c$ appear consecutively in this step.  We maintain the size, $S[c]$, of color class $c$ and a doubly-linked list of the elements corresponding to color $c$, $D[c]$.  We also maintain an array of pointers, $P[v]$, to $v$'s entry on its color list,
$D[C[v]]$.

If there are $S[c]$ rows in $M$ starting with $c$, and all of these are identical
-- except for the rows' last coordinates, which are the vertices being colored -- then color class $c$ is unchanged.  Otherwise, it is broken into multiple pieces, of sizes, say, $p_1 \geq p_2 \geq \dots \geq p_r \geq 1$. Note that if not all elements of $c$ were $L$-neighbors, then one of these pieces corresponds to the vertices of color $c$ that were not $L$-neighbors.  If this set of non-neighbors is the largest sub-piece, then we do not have to visit its members.  We simply update $S[c]$ and $D[c]$ by
decrementing $S[c]$ and deleting $v$'s entry in $D[c]$ for each $v$ in any of the smaller pieces.
Since $D[c]$ is a doubly-linked list and we maintain the pointer array $P[v]$ to $v$'s entry, 
this takes time $O(1)$ for each such vertex.
Thus, the time for processing color class $c$ is
$O(p_2 + \cdots + p_r)$.

Using Claim \ref{claim:logn}, each vertex appears as an $L$-vertex at most $1 + \log n$ times.
Since the time it contributes to that round is at most $O(n)$, it follows that the total time for the algorithm is at most $O(n^2\log n)$.
\end{proof}

\subsection{Example Run}
For instance, consider the algorithm run on the following graph.

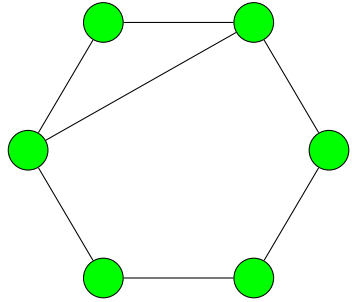
\begin{center}
\begin{tikzpicture}
 
   \tikzstyle{every node}=[draw,circle,fill=green,minimum size=15pt,
                            inner sep=0pt]
   
    \draw (-1,3.4) node (1) {};
		\draw (1,3.4) node (2) {};
		\draw (2,1.7) node (3) {};
		\draw (1,0) node (4) {};
		\draw (-1,0) node (5) {};
		\draw (-2,1.7) node (6) {};

		\draw (1) -- (2);
		\draw (2) -- (3);
		\draw (3) -- (4);
		\draw (4) -- (5);
		\draw (5) -- (6);
		\draw (6) -- (1);
		\draw (2) -- (6);
		
\end{tikzpicture}
\end{center}
Initially the graph is monochromatic, so we can take the initial color to be $g$ for all vertices (and so, $L$ is initialized to be $\{g\}$). Each vertex has either one or two uncolored neighbors, so in the second iteration, there will only be two new colors, corresponding to the tuples $(g, \{g, g\})$ and $(g, \{g, g, g\})$. Representing them by green and yellow, we obtain the following colored graph after the $1$st iteration (observe that all vertices had to be updated, since each vertex had at least one uncolored neighbor).
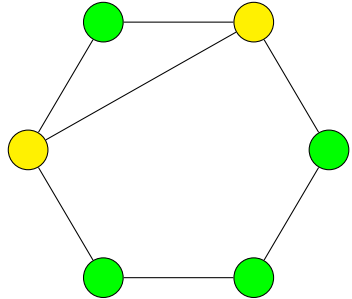
\begin{center}
\begin{tikzpicture}
 
    \tikzstyle{every node}=[draw,circle,fill=yellow,minimum size=15pt,
                            inner sep=0pt]
    \draw (1,3.4) node (2) {};
		\draw (-2,1.7) node (6) {};
		
		\tikzstyle{every node}=[draw,circle,fill=green,minimum size=15pt,
                            inner sep=0pt]
    \draw (-1,3.4) node (1) {};
		\draw (2,1.7) node (3) {};
		\draw (1,0) node (4) {};
		\draw (-1,0) node (5) {};

		\draw (1) -- (2);
		\draw (2) -- (3);
		\draw (3) -- (4);
		\draw (4) -- (5);
		\draw (5) -- (6);
		\draw (6) -- (1);
		\draw (2) -- (6);
		
\end{tikzpicture}
\end{center}
Renaming these colors $g$ and $y$, observe that now $L = \{y\}$ (since the largest part of the partition corresponded to the green vertices), and is in particular nonempty. So we keep going. In the next iteration, we can ignore the vertex on the bottom right, since it is not adjacent to any yellow vertex. Updating the other vertices, note that there are three new colors. Denoting the updated new color classes as
\begin{align*}
b &:= (g, \{y, y\}) \\
g &:= (g, \{y\}) \\
y &:= (y, \{y\})\\
p &:= (g), 
\end{align*}
our graph in the next iteration looks as follows.
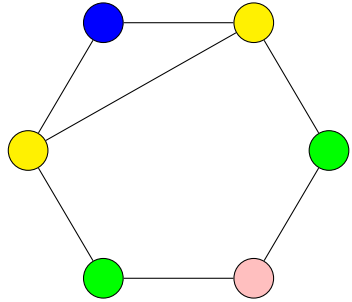
\begin{center}
\begin{tikzpicture}
 
    \tikzstyle{every node}=[draw,circle,fill=yellow,minimum size=15pt,
                            inner sep=0pt]
    \draw (1,3.4) node (2) {};
		\draw (-2,1.7) node (6) {};
		
		\tikzstyle{every node}=[draw,circle,fill=blue,minimum size=15pt,
                            inner sep=0pt]
    \draw (-1,3.4) node (1) {};
		
		\tikzstyle{every node}=[draw,circle,fill=pink,minimum size=15pt,
                            inner sep=0pt]
    \draw (1,0) node (4) {};
		
		\tikzstyle{every node}=[draw,circle,fill=green,minimum size=15pt,
                            inner sep=0pt]
    \draw (2,1.7) node (3) {};
		\draw (-1,0) node (5) {};

		\draw (1) -- (2);
		\draw (2) -- (3);
		\draw (3) -- (4);
		\draw (4) -- (5);
		\draw (5) -- (6);
		\draw (6) -- (1);
		\draw (2) -- (6);
		
\end{tikzpicture}
\end{center}
Now, consider how to update $L$. The old color classes were green and yellow. The yellow vertices have not been refined, so we do not need to include them in $L$. The green vertices have now been partitioned into three parts, colored green, blue and pink. Of these, the largest one remains green, so we can ignore it, and include the two others, so that now $L = \{b, p\}$, and is still nonempty.

Consider the next update. The only vertices that need to be updated are the yellow ones (for being adjacent to the lone blue vertex), which both keep the same color, and the green ones (for being adjacent to the lone pink vertex), which will both now keep the same color.

Thus, there is no change in this round, so $L = \varnothing$ and the algorithm is complete.

\section{The Higher Dimensional Algorithm}

\subsection{Description of the Algorithm}

When $k \geq 2$, the algorithm and its analysis are essentially the same, with a few added subtleties. Once again, we start with any undirected, uncolored graph $G = (V, E)$. We are now concerned with $k$-tuples, i.e. members of $V^k$. Let's define the \emph{neighbor} of such a $k$-tuple.
\begin{defn}
Let $\vec{x} = (x_1, \ldots, x_k) \in V^k$, $y \in V$, and $1 \leq j \leq k$. Then, let $\vec{x}[j, y] \in V^k$ denote the $k$-tuple obtained from $\vec{x}$ by replacing $x_j$ by $y$. The $k$-tuples $\vec{x}[j, y]$ and $\vec{x}$ are said to be \emph{$j$-neighbors} for any $y \in V$. We also say $\vec{x}[j, y]$ is the \emph{$j$-neighbor of $\vec{x}$ corresponding to $y$}.
\end{defn}

We define the initial coloring of all $k$-tuples to correspond to encodings of their isomorphism \emph{types}. Precisely speaking, we define $\prescript{}{0}\chi_G^k(\vec{v})$ to be the (ordered) isomorphism class of $G[v_1, \ldots, v_k]$; that is, $\prescript{}{0}\chi_G^k(\vec{v}) = \prescript{}{0}\chi_G^k(\vec{w})$ if and only if the map $v_i \mapsto w_i$ is an isomorphism. As before, we maintain the work list $L$ that stores all the color classes \emph{updated} during the previous iteration. We now present the complete algorithm below.

\IncMargin{1em}
\RestyleAlgo{boxruled}
\begin{algorithm}[h]\label{alg:WLk}
\SetAlgoLined
\SetKwInOut{Input}{Input}\SetKwInOut{Output}{Output}
\Input{An uncolored, undirected graph $G = (V, E)$.}
\Output{The \emph{stable} coloring of $V^k$.}

\vspace*{.1in}

\textbf{Initialization:} $C[\vec{v}] = \prescript{}{0}\chi_G^k(\vec{v})$, for all $\vec{v}\in V^k$; $M = \varnothing$; $L = \set{\prescript{}{0}\chi_G^k(\vec{v}) : \vec{v} \in V^k}$.  

** All tuples initially colored with their \emph{isomorphism types}; multiset $M$ empty; work list
$L$ initialized with the set of initial colors of $k$-tuples in $G$.

\vspace*{.1in}

\While{$L \neq \varnothing$}{
  \For{each tuple color class, $c \in L$}{remove $c$ from $L$.

    \For{each tuple $\vec{w}$ with $C[\vec{w}] = c$}{
      \For{each $j \leq k$}{
			\For{each $u \in V$}{
			{let $\vec{v} = \vec{w}[j, u]$\;
			add $(\vec{v}, C[\vec{v}[1, u]], \ldots, C[\vec{v}[k, u]])$ to $M$.}}
     }
   }
	}

Perform Radix Sort of $M$.

Scan $M$ replacing tuples $(\vec{v},c^1_1, c^1_2, \ldots, c^1_k), \ldots, (\vec{v},c^r_1, c^r_2, \ldots, c^r_k)$, with the single tuple: $(C[\vec{v}]; (c^1_1, c^1_2, \ldots, c^1_k), \ldots, (c^r_1, c^r_2, \ldots, c^r_k); \vec{v})$.

Perform Radix Sort of $M$.

Scan $M$ for each color class $c$ that has been split; leave the largest part still colored $c$, and update the colors of the other parts of $c$; add these new colors to $L$.
}
Output the current coloring of all $k$-tuples in $V^k$.
\caption{The $k$-dimensional Weisfeiler-Leman Algorithm}
\end{algorithm}

Once again, we can define an \emph{$L$-tuple} as a $k$-tuple whose color class is in $L$. We can
also talk about an \emph{$L$-neighbor} of a $k$-tuple, $\vec{v}$, which is simply an $L$-tuple that is a
$j$-neighbor of $\vec{v}$ for some $j\leq k$. 

\begin{remark}
It is worth pointing out the similarity between this algorithm and the one-dimensional version, particularly in the two sorting steps, in lines 15 and 17. Once again, the sort in line 15 is indexed by the tuples themselves, and so for each $\vec{v}$ it clumps together the color classes of its $L$-neighbors with multiplicity, with the purpose once again being to determine a canonical, well-defined label for the new color classes of the tuples. The sort in line 17 is indexed by the \emph{old color classes} of the tuples, and so for each old color class it clumps together all tuples that used to be in it. This now enables us to count the sizes of the new color classes to determine which (if any) has been split, in order to update $L$. Again, bounding this time is crucial to the eventual analysis.
\end{remark}

\subsection{Proof of Correctness and Runtime}

The analysis for the higher dimensional version of the algorithm is similar to the one-dimensional one.

\begin{claim}
Algorithm \ref{alg:WLk} terminates with a stable coloring of $V^k$.
\end{claim}
\begin{proof}
This proof is by and large the same as before.
\end{proof}

\begin{claim}\label{claim:logndimk}
The color class of any $k$-tuple $\vec{v} \in V^k$ can appear in $L$ at most $O(k\log n)$ times.
\end{claim}
\begin{proof}
This is also similar to the one-dimensional case. A $k$-tuple $\vec{v} \in V^k$ will have its color appear in $L$ during the $(i+1)$st round of the algorithm only if its color class was just split, and $\vec{v}$ was \emph{not} in the largest piece of the ensuing partition.  So each time $\vec{v}$'s color class appears on $L$, this class is at most half the size it was during the previous round. There are $n^k$ $k$-tuples in all, and so any particular $k$-tuple $\vec{v}$ can have its color class treated at most $\log(n^k)$ times.
\end{proof}

\begin{claim}\label{claim:runtimedimk}
On input $G = (V, E)$ with $|V| = n$, Algorithm \ref{alg:WLk}  runs in time $O(k^2n^{k + 1}\log n)$.
\end{claim}
\begin{proof}
Consider the main while-loop (lines 3-19). The innermost for-loop (lines 8-11) takes time $O(n)$ to iterate over $u \in V$ and update $M$. The for-loop in lines 7-12 therefore requires time $O(kn)$. This is done for each $L$-tuple in that round, accounting for the for-loop in lines 4-14. The radix sorts on lines 15 and 17, as well as the scanning and updating steps on line 16, as before, are all linear in the size of $M$. Line 18 is also implemented exactly as before, with the aid of the array $S[c]$ of color class sizes and the doubly-linked list $D[c]$ of elements within each color class, together with the pointers $P[\vec{v}]$. The updating process is precisely as before, and so the total processing time is still linear in the size of $M$. But note that $M$ has one entry for each neighbor of an $L$-tuple, and its size, therefore, is also bounded by $O(kn)$ times the number of $L$-tuples. It remains now to verify the number of times a given tuple can appear in $L$, which we know is $O(k\log n)$ from Claim \ref{claim:logndimk}.

There are $n^k$ $k$-tuples in total, and each of them appears as an $L$-tuple (and therefore gets its color class treated) at most $O(k\log n)$ times, with each such treatment taking $O(kn)$ time, so that the total complexity is $O(n^k \cdot k\log n \cdot kn) = O(k^2n^{k + 1}\log n)$, as desired.
\end{proof}


\begin{thebibliography}{CFI92}
\bibitem[AHU74]{AHU} Alfred V. Aho, John E. Hopcroft, and Jeffrey D. Ullman, \emph{The Design and Analysis of Computer Algorithms}, Addison-Wesley Longman Publishing Co., Inc., Boston, MA (1974).
\bibitem[BK80]{BK} Laszlo Babai and Ludik Ku\v cera, ``Canonical Labelling of Graphs in
    Linear Average Time,'' {\it 20th IEEE Symp. on Foundations of Computer
    Science} (1980), 39-46.
\bibitem[BBG15]{BBG} Christoph Berkholz, Paul Bonsma, and Martin Grohe, ``Tight Lower and Upper Bounds for the Complexity of Canonical Colour Refinement,'' arXiv:1509.08251v1 [cs] (2015).
\bibitem[CFI92]{CFI} Jin-Yi Cai, Martin F\"urer, and Neil Immerman, ``An Optimal Lower Bound on the Number
    of Variables for Graph Identification,'' {\it Combinatorica} {\bf 12}
    (4) (1992) 389-410.
\commentout{\bibitem[CLRS09]{CLRS} Thomas H. Cormen, Charles E. Leiserson, Ronald L. Rivest, and Clifford Stein, {Introduction to Algorithms}, Third Ed., the MIT Press, Cambridge, MA (2009).}
\bibitem[IL90]{canon} Neil Immerman and Eric S. Lander, ``Describing Graphs: A First-Order
    Approach to Graph Canonization,'' in {\it
    Complexity Theory Retrospective,} Alan Selman, ed., 
    Springer-Verlag (1990), 59-81.
\end{thebibliography}
\end{document}